\documentclass[journal]{IEEEtran}
\usepackage{amsmath,amsfonts}
\usepackage{algorithmic}
\usepackage{algorithm}
\usepackage{array}
\usepackage[caption=false,font=normalsize,labelfont=sf,textfont=sf]{subfig}
\usepackage{textcomp}
\usepackage{stfloats}
\usepackage{url}
\usepackage{verbatim}
\usepackage{graphicx}
\usepackage{cite}

\usepackage[T1]{fontenc}

\usepackage{amssymb}

\newtheorem{lemma}{Lemma}

\begin{document}

\title{Identifiability in Blind Source Separation\\ through Stabilizer Shrinkage:\\
Unifying Non-Gaussianity and Observation Diversity}

%\author{IEEE Publication Technology,~\IEEEmembership{Staff,~IEEE,}
        % <-this % stops a space
%\thanks{This paper was produced by the IEEE Publication Technology Group. They are in Piscataway, NJ.}% <-this % stops a space
%\thanks{Manuscript received April 19, 2021; revised August 16, 2021.}}
\author{
Tomomi Ogawa ~%\IEEEmembership{Member,~IEEE}
and~Hiroki Matsumoto%
\thanks{
Tomomi Ogawa is with the Graduate School of Science and Engineering,
Tokyo Denki University, Saitama, Japan
(e-mail: to.ogawa@mail.dendai.ac.jp).
}%S
\thanks{
Hiroki Matsumoto is with the Graduate School of Engineering,
Maebashi Institute of Technology, Gumma, Japan
(e-mail: matsumoto@maebashi-it.ac.jp).
}%
}%


\maketitle

\begin{abstract}
Identifiability is a central issue in blind source separation (BSS), determining
whether latent sources can be uniquely recovered from observed mixtures.
Classical approaches address identifiability either by exploiting source
non-Gaussianity via higher-order statistics (HOS) or by enriching the observation
structure through temporal, spatial, or multi-channel diversity using
second-order statistics (SOS), and these routes are often regarded as
fundamentally different.
In this paper, we revisit identifiability in BSS from a structural perspective,
interpreting it as constraint-induced reduction of residual ambiguity in the
mixing model.
Within this framework, the observation mechanism is viewed broadly to include
both input-side statistical constraints and output-side observation structures.
HOS-based and SOS-based approaches are then unified as mechanisms of
\emph{stabilizer shrinkage}, in which observation-induced constraints reduce an
initially continuous ambiguity to a finite residual one.
To connect this structural viewpoint with finite-sample regimes, we introduce a
Jacobian-based sensitivity probe as a numerical diagnostic of local
identifiability.
Numerical experiments show that increasing non-Gaussianity or observation
diversity suppresses the same residual symmetry, revealing a structural
trade-off between source statistics and observation design.
These results provide a unified interpretation of classical BSS methods and
clarify how observation constraints govern identifiability.
\end{abstract}


% ===== Main sections =====
\section{Introduction}
\label{sec:introduction}
Blind source separation (BSS) aims to recover latent source
signals from observed mixtures without prior knowledge of
the mixing process, and identifiability lies at the core of
this task~\cite{kofidis_blind_2016}.
Even when separation algorithms converge and standard performance indices appear
satisfactory, residual ambiguities may remain~\cite{tong_indeterminacy_1991}.
Such ambiguities reflect intrinsic limitations of the observation mechanism
rather than algorithmic failure.
Clarifying when and why they can or cannot be eliminated is therefore essential
for interpreting classical BSS methods and for guiding the design of new ones.

Classical studies of identifiability in BSS have largely followed two
complementary routes~\cite{comon_independent_1994,hyvarinen_independent_2000,belouchrani_blind_1997,tong_amuse_1990}.
Higher-order statistics (HOS)--based methods exploit deviations from Gaussianity to
impose distributional constraints that are not accessible through second-order
analysis~\cite{comon2010handbook}.
In contrast, second-order statistics (SOS)--based methods enrich the observation
structure via temporal, spatial, or multi-channel diversity, thereby imposing
multiple correlation-based constraints~\cite{tong_indeterminacy_1991,aissa-el-bey_general_2008}.
Throughout this paper, we use the term \emph{multi-SOS} to refer to second-order
statistics constructed from multiple lags or blocks (e.g., delayed covariance
matrices in SOBI-type methods), in contrast to single-lag second-order statistics.
These two routes are commonly regarded as fundamentally different, resulting in
distinct algorithmic families and design heuristics.
Consequently, identifiability has typically been discussed either from the
standpoint of source statistics or from that of observation diversity, with
limited conceptual integration between the two.

Recent data-driven and deep learning--based approaches further complicate this
landscape~\cite{kivva_identifiability_2022,helwani_sound_2024}.
Although such methods often achieve strong empirical performance,
they do not directly address the structural question of identifiability:
to what extent does the observation mechanism itself resolve ambiguity,
independently of model capacity or training data?
This question motivates a structural treatment of identifiability based on
how observations constrain admissible transformations of the mixing model~\cite{mika_lie_2020}.

From this perspective, identifiability is not viewed as an algorithmic outcome
but as a property of the observation mechanism itself.
Specifically, it depends on whether observation-induced constraints reduce an
initially continuous ambiguity to a finite residual one.
In the sequel, this idea is formalized by introducing an observation-induced
equivalence relation in Section~\ref{sec:problem-formulation}, characterizing the
associated stabilizer of admissible transformations in
Section~\ref{sec:structural-viewpoint}, and examining its local manifestation
through Jacobian sensitivity in Section~\ref{sec:local-identifiability}.

Within this framework, observation statistics are interpreted as invariance
constraints on the mixing matrix, and residual ambiguities correspond to
transformations that preserve all enforced statistics.
Identifiability is therefore characterized by the extent to which these
constraints restrict the admissible transformation group.
In particular, HOS-based approaches impose constraints primarily on the
\emph{input side} through assumptions on source statistics such as
non-Gaussianity and independence,
whereas multi-SOS-based approaches impose constraints on the \emph{output side}
by enriching the observation structure through temporal, spatial, or
multi-channel diversity in second-order statistics.
Both classes of methods can thus be interpreted as mechanisms that reduce
residual ambiguity by shrinking the stabilizer of the mixing model.
This structural reduction, termed \emph{stabilizer shrinkage} and formalized in
Section~\ref{subsec:constraint-accumulation}, provides a unified viewpoint in
which distributional information and observation diversity act as alternative
sources of constraint richness for ambiguity reduction.

An immediate implication of this viewpoint is that observation mechanisms
producing identical stabilizer shrinkage necessarily yield identical local
identifiability properties, regardless of whether the constraints arise from
HOS or from multi-SOS statistics.
Differences between algorithms therefore reflect how stabilizer-reducing
constraints are instantiated rather than fundamentally distinct identifiability
principles.

To connect this structural analysis to finite-sample regimes, we introduce a
Jacobian-based sensitivity probe as a numerical diagnostic of local
identifiability.
The probe is used strictly as a measurement tool, not as a new identifiability
criterion, to assess how strongly observation constraints suppress weakly
observable perturbation directions.
Because the Jacobian captures the first-order response of the observation map,
its nullspace reflects residual stabilizer directions, enabling diagnostic
comparison of different constraint configurations even when global
identifiability conditions are formally satisfied.

Controlled numerical experiments illustrate that (i) HOS constraints strengthen
identifiability as non-Gaussianity increases, (ii) multi-lag SOS constraints
recover identifiability through observation diversity even under Gaussian
sources, and (iii) the trade-off between these mechanisms can be visualized
within a common structural framework.

The main contributions of this paper are summarized as follows:
\begin{itemize}
\item A structural interpretation of identifiability in BSS based on
constraint-induced ambiguity reduction, independent of specific separation
algorithms.
\item A unified framework of stabilizer shrinkage that relates HOS-based and
SOS-based approaches and clarifies their complementary roles.
\item A Jacobian-based sensitivity probe as a practical diagnostic tool for
visualizing local identifiability under finite-sample conditions.
\item Numerical experiments illustrating how non-Gaussianity and observation
diversity contribute to ambiguity reduction within a common structural probe.
\end{itemize}

The remainder of the paper is organized as follows.
Section~\ref{sec:problem-formulation} introduces the problem formulation and inherent ambiguities.
Section~\ref{sec:background} reviews classical HOS-based and SOS-based approaches from a
constraint-oriented perspective.
Section~\ref{sec:structural-viewpoint} develops the structural viewpoint on observation-induced constraints
and symmetry.
Section~\ref{sec:local-identifiability} formalizes local identifiability via Jacobian sensitivity.
Section~\ref{sec:tradeoff} examines the relationship between non-Gaussianity and observation diversity,
and Section~\ref{sec:experiments} presents numerical experiments validating the proposed
interpretation.
Section~\ref{sec:conclusion} concludes the paper and discusses implications for the design
and diagnosis of observation constraints in blind source separation.
Appendix~\ref{app:jacobian-probe} provides the formal definition of the Jacobian-based probe, while
Appendix~\ref{app:asymmetry} summarizes representative counterexamples in which observation-induced
constraints fail to reduce ambiguity.
\section{Problem Formulation and Notation}
\label{sec:problem-formulation}
%--------------------
\subsection{Blind Source Separation Model}
\label{subsec:bss-model}

We begin by fixing the basic blind source separation (BSS) model and the
associated equivalence class, which serve as the reference throughout the
paper.

We consider the standard linear instantaneous BSS model~\cite{kofidis_blind_2016,pedersen_survey_2007}
\begin{equation}
x(t) = H s(t),
\label{eq:bss-model}
\end{equation}
where $s(t)\in\mathbb{R}^n$ denotes statistically independent source signals and
$H\in\mathbb{R}^{m\times n}$ is an unknown mixing matrix.
The objective of BSS is to recover the source signals, or an equivalent
representation thereof, solely from the observed mixtures $x(t)$, without prior
knowledge of $H$.

Throughout this paper, we focus on square or overdetermined mixtures
($m\geq n$), and consider both real-valued and complex-valued formulations, as
they naturally arise in applications such as audio signal processing and
communication systems~\cite{congedo_blind_2008}.
No specific algorithmic structure is assumed at this stage; the analysis is
intended to be independent of particular separation methods~\cite{comon2010handbook}.

When second-order statistics are involved, we assume that the observations are
whitened using the sample covariance~\cite{comon_independent_1994,hyvarinen_independent_2000}.
Here, for a zero-mean vector process $v(t)$, we define
$\mathrm{Cov}(v):=\mathbb{E}[\,v(t)v(t)^\top\,]$
(or $\mathbb{E}[\,v(t)v(t)^{\ast}\,]$ in the complex-valued case),
where $\mathbb{E}[\cdot]$ denotes statistical expectation, $(\cdot)^\top$ the
transpose, and $(\cdot)^{\ast}$ the conjugate transpose.

Under the linear model \eqref{eq:bss-model}, the covariance of the observations
satisfies
\begin{equation}
\mathrm{Cov}(x) = H\,\mathrm{Cov}(s)\,H^\top .
\label{eq:covariance-relation}
\end{equation}

In practice, whitening is achieved by a linear transformation $W$ such that
\begin{equation}
\mathrm{Cov}(W x) = I ,
\label{eq:whitening}
\end{equation}
where $I$ denotes the identity matrix of appropriate dimension.
Under \eqref{eq:bss-model}, this implies that the effective mixing matrix $W H$
satisfies $(W H)(W H)^\top = I$ in the real-valued case (or
$(W H)(W H)^{\ast} = I$ in the complex-valued case), and can therefore be regarded
as orthogonal (or unitary).
In the following, we adopt the common convention of absorbing the whitening
transform into the observation, and reuse the symbol $x$ to denote the whitened
signal.

Accordingly, without loss of generality, the mixing matrix can be considered up
to a residual right-orthogonal (or unitary) transformation.
Specifically, one may assume that after whitening,
\begin{equation}
\mathrm{Cov}(x)=I
\;\Longrightarrow\;
H \sim H Q ,
\label{eq:whitened-equivalence}
\end{equation}
where the symbol $\sim$ denotes an equivalence relation on the set of mixing
matrices, identifying representations that differ only by a residual
right-orthogonal (or unitary) transformation.
The residual transformation $Q$ satisfies
\begin{equation}
Q \in
\begin{cases}
\mathcal{O}(n), & \text{real-valued case},\\
\mathcal{U}(n), & \text{complex-valued case},
\end{cases}
\label{eq:orthogonal-unitary}
\end{equation}
reflecting the fact that whitening fixes the zero-lag covariance up to an
orthogonal or unitary change of basis.
Here, $\mathcal{O}(n)$ and $\mathcal{U}(n)$ denote the orthogonal and unitary
groups of dimension $n$, respectively, with dimensions consistent with the
whitened model.
The symbol $I$ should not be confused with the index set $\mathcal{I}$ used later
to label observation constraints.

This preprocessing does not alter identifiability in the sense considered here.
Rather, it isolates the residual ambiguities to orthogonal (or unitary)
transformations, a property that plays a central role in the subsequent analysis.

It is well known that, even under ideal conditions, the BSS problem does not
admit a unique solution.
Instead, recovered sources are determined only up to inherent ambiguities.
In particular, demixing solutions related by permutation and component-wise
scaling (or phase rotations in the complex-valued case) are indistinguishable
from the observations~\cite{comon_independent_1994}.
Accordingly, throughout this paper we regard two mixing matrices as equivalent if
\begin{equation}
H \;\sim\; H P D,
\label{eq:bss-equivalence}
\end{equation}
where $P$ is a permutation matrix and $D$ is a nonsingular diagonal matrix
(representing scaling or phase factors).
This equivalence relation fixes the notion of identifiability considered in the
sequel.
%--------------------
\subsection{Ambiguities and Equivalence up to Inherent Transformations}
\label{subsec:ambiguity-equivalence}

A fundamental characteristic of blind source separation is that its solutions
are not unique, even in the absence of noise and estimation errors.
This non-uniqueness does not arise from algorithmic imperfections, but is
intrinsic to the BSS model itself~\cite{benveniste_robust_1980,tong_indeterminacy_1991}.
Specifically, certain transformations applied to the recovered sources produce
alternative solutions that remain consistent with the observed mixtures.

In practice, the inherent ambiguities described by
\eqref{eq:bss-equivalence} include permutations of the source order and
component-wise scaling, or phase rotations in the complex-valued case.
Such transformations preserve the validity of the separation in the sense that
they lead to reconstructed sources compatible with the observed data, despite
differences in parameter representation.
Throughout this paper, equivalence is understood strictly in this BSS-specific
sense.
At this stage, no additional algebraic or group-theoretic structure is assumed;
the purpose of this subsection is solely to clarify which ambiguities are
considered inherent to the model.

This perspective is essential for interpreting identifiability in blind source
separation.
Rather than asking whether a unique parameter estimate exists, the relevant
question is which ambiguities can or cannot be resolved by the available
information~\cite{bellman_structural_1970}.
Identifiability is therefore determined by the extent to which observation-derived
constraints restrict the equivalence class \eqref{eq:bss-equivalence}
~\cite{tong_indeterminacy_1991}.

In the following subsection, this viewpoint is formalized by examining how
different classes of observation statistics act as constraints that progressively
limit the set of admissible transformations.

\subsection{Observation Statistics as Constraint-Inducing Information}
\label{subsec:obs-constraints}

Throughout this paper, identifiability is considered up to the \emph{finite}
inherent ambiguities described in Section~\ref{subsec:ambiguity-equivalence}
(permutation and component-wise scaling, or phase rotations in the
complex-valued case).
After whitening (Section~\ref{subsec:bss-model}), these finite ambiguities are
preceded by a \emph{continuous} right-orthogonal (or unitary) freedom, reflecting
the fact that second-order statistics alone cannot distinguish right-unitary
reparameterizations of the mixing matrix.
The general linear group $\mathrm{GL}(n)$ is introduced here as a convenient
parametrization of admissible reparameterizations of the mixing matrix
\emph{before} observation-induced constraints are imposed; after whitening,
the relevant residual transformations reduce locally to right-orthogonal
(or unitary) actions as discussed in Section~\ref{subsec:bss-model}.

In blind source separation, the only information available for resolving these
ambiguities arises from statistical properties of the observed signals.
We therefore interpret observation statistics structurally, as constraints
imposed by the observation mechanism itself, which restrict admissible
transformations of the mixing matrix while remaining consistent with the
observed data.

To formalize this viewpoint, we introduce an observation map
\begin{equation}
\Phi:\; H \;\longmapsto\; \bigl\{ A_i(H) \bigr\}_{i \in \mathcal{I}},
\label{eq:obs-map}
\end{equation}
where each $A_i(H)$ is a matrix-valued quantity induced by a chosen class of
observation statistics.
The index set $\mathcal{I}$ specifies which statistical constraints are imposed
simultaneously.
Typical examples include higher-order cumulant matrices in HOS-based approaches
or covariance-based operators in SOS-based approaches~\cite{cardoso_blind_1993,belouchrani_blind_1997}.
Each operator $A_i(H)$ is defined by applying a fixed, predetermined
construction rule (e.g., a population cumulant or covariance operator) to the
model $x(t)=Hs(t)$; the same construction rule is applied when evaluating
$A_i(HG)$ for any admissible transformation $G$.
For example, in the HOS case $A_i(H)$ may correspond to population fourth-order
cumulant matrices, whereas in the SOS case it may represent lagged covariance
operators evaluated at different lags.

For clarity, the term \emph{multi--SOS} is used to collectively refer to
second-order constraint families involving multiple operators, including but
not limited to lagged covariance matrices, multi-sensor or multi-channel
observations, block or oversampled covariance structures, and covariance
operators obtained under multiple experimental conditions~\cite{belouchrani_blind_1997,moulines_subspace_1995,tong_blind_1994}.
In this usage, specific instances such as \emph{multi-lag SOS} correspond to
concrete realizations based on temporal lag diversity.
In contrast, the labels \emph{HOS-based} and \emph{SOS-based} are used when
emphasizing the type of statistics being exploited, whereas the term
\emph{multi--SOS} highlights the structural level at which multiple second-order
operators jointly impose constraints.

The collection $\{A_i(H)\}_{i \in \mathcal{I}}$ represents invariants of the
observation: any admissible transformation of $H$ must preserve all of these
quantities.
Specifically, a transformation $G \in \mathrm{GL}(n)$ is said to be compatible
with the observation if
\begin{equation}
A_i(HG) = A_i(H),
\qquad \forall\, i \in \mathcal{I}.
\label{eq:constraint-invariance}
\end{equation}

This invariance condition can be equivalently rewritten as
\begin{equation}
\begin{aligned}
A_i(HG)=A_i(H)\ \forall i
\quad &\Longleftrightarrow \quad
\Phi(HG)=\Phi(H),
\end{aligned}
\label{eq:obs-kernel}
\end{equation}
which makes explicit that admissible transformations are those that leave the
observation map unchanged.

Accordingly, we define the residual ambiguity set induced by the observation map
$\Phi$ as
\begin{equation}
\mathrm{Stab}(\mathcal{I})
\;\triangleq\;
\bigl\{\, G \in \mathrm{GL}(n)
\;\big|\;
A_i(HG) = A_i(H),\ \forall\, i \in \mathcal{I}
\bigr\}.
\label{eq:global-ambiguity}
\end{equation}

The terminology ``stabilizer'' is used here in the standard sense of group
actions, referring to the subgroup of transformations that leave a given object
invariant (see, e.g.,~\cite{Herstein1996AbstractAlgebra}).

The set $\mathrm{Stab}(\mathcal{I})$ captures the remaining symmetries of the
mixing matrix that cannot be resolved from the available statistics.
In practice, different choices of observation statistics lead to different
\emph{degrees of stabilizer reduction}: while whitening alone leaves a
continuous right-unitary stabilizer, appropriately chosen and sufficiently rich
observation constraints can, under well-understood conditions, reduce this
continuous ambiguity to a \emph{finite} residual group.
In this sense, identifiability is governed by how strongly the observation map
$\Phi$ restricts admissible transformations of $H$, rather than by the numerical
procedure used to estimate it.

\paragraph*{Remark on identifiability}
Identifiability in blind source separation does not arise from output
non-Gaussianity alone, nor from structural richness of the input process when
such structure is not fixed by the observation operators.
Instead, identifiability emerges only when source-side constraints
(e.g., non-Gaussianity combined with independence) or output-side operator
diversity effectively reduce the continuous right-unitary freedom to a finite
residual ambiguity.
Formal counterexamples illustrating these asymmetries are summarized in
Appendix~\ref{app:asymmetry}.
%-----------------------------------------------

\subsection{Scope and Assumptions}
\label{subsec:scope}

The scope of this paper is deliberately restricted in order to focus on the
structural aspects of identifiability in blind source separation.
Specifically, we consider linear instantaneous mixture models with statistically
independent sources~\cite{naanaa_geometric_2012}, as introduced in
Section~\ref{subsec:bss-model}.
This setting encompasses a broad class of classical BSS problems while allowing
the effects of different observation constraints to be examined without
confounding factors arising from model complexity.

The analysis primarily addresses separation methods based on second-order and
higher-order observation statistics, including those exploiting temporal,
spatial, or multi-observation diversity.
No assumptions are made on a particular algorithmic implementation, and the
discussion is not tied to specific optimization procedures.
Instead, the emphasis is placed on how different classes of statistical
information constrain the underlying model and influence the resolution of
inherent ambiguities.

Several extensions of the basic BSS model are intentionally not considered~\cite{kivva_identifiability_2022,helwani_sound_2024,gong_double_2018}.
These include nonlinear or convolutive mixtures, deep learning--based separation
frameworks, and problem formulations that rely on supervised or semi-supervised
training.
While these settings are of practical interest, they introduce additional
sources of ambiguity and inductive bias that obscure the structural effects
studied here.
By limiting attention to the classical linear BSS framework, the analysis aims to
isolate and clarify the role of observation-induced constraints in shaping
identifiability.
\section{Background: HOS-Based and SOS-Based Blind Source Separation}
\label{sec:background}

Blind source separation has been studied extensively under a variety of
assumptions and problem settings.
Among classical approaches, methods based on higher-order statistics (HOS)
and those based on second-order statistics (SOS) have played a central role
and are often treated as distinct methodological families.
This section briefly reviews these two classes, focusing on their roles as
\emph{constraint-inducing mechanisms} rather than on algorithmic details.

%--------------------------------
\subsection{HOS-Based Approaches}
HOS-based approaches exploit the non-Gaussianity of source signals by
introducing constraints derived from higher-order moments or cumulants,
which are invisible to second-order analysis~\cite{
comon_independent_1994,
hyvarinen_independent_2000,
cardoso_blind_1993,
amari_new_1995,
comon2010handbook,
noauthor_handbook_nodate,
shalvi_new_1990,
cardoso_high-order_1999}.

From a structural viewpoint, HOS-based methods impose constraints of the form
\begin{equation}
A_i(HG) \;=\; A_i(H),
\qquad i \in \mathcal{I}_{\mathrm{HOS}},
\label{eq:hos-invariance}
\end{equation}
where $A_i(H)$ denotes a matrix induced by higher-order statistics
(e.g., cumulant matrices),
$G \in \mathrm{GL}(N)$ represents an admissible transformation of the mixing
matrix, and $\mathcal{I}_{\mathrm{HOS}}$ indexes the chosen set of HOS constraints.
Equation~\eqref{eq:hos-invariance} expresses the fact that only transformations
preserving higher-order statistical structure remain compatible with the
observation.

A key advantage of HOS-based methods is their ability to eliminate ambiguities
that persist under purely second-order analysis, particularly when source
distributions deviate sufficiently from Gaussianity.
At the same time, the effectiveness of these constraints depends on the
reliability of higher-order statistical estimation, which can be sensitive to
finite-sample effects and noise.

%--------------------------------
\subsection{SOS-Based Approaches}
In contrast, SOS-based approaches rely on structural diversity in the
observations, without assuming specific source distributions~\cite{
belouchrani_blind_1997,
tong_amuse_1990,
tong_indeterminacy_1991,
moulines_subspace_1995,
aissa-el-bey_general_2008,
tong_new_1991,
tong_blind_1994}.
By exploiting multiple time lags, sensor arrays, or observation blocks,
they introduce collections of second-order constraints.

A typical SOS formulation considers a collection of covariance matrices
\begin{equation}
R_z(\tau) \;=\; \mathbb{E}\!\left[ z(t)\,z(t-\tau)^{\top} \right],
\qquad \tau \in \mathcal{I}_{\mathrm{SOS}},
\label{eq:sos-covariances}
\end{equation}
where $z(t)$ denotes the observed (often whitened) signal,
and $\mathcal{I}_{\mathrm{SOS}}$ specifies a set of nonzero lags or observation
conditions.
Each lag $\tau$ induces a second-order constraint, and separation is achieved
by enforcing joint compatibility across all selected lags.

Structurally, SOS-based methods restrict admissible transformations by
requiring simultaneous consistency with multiple second-order relations.
While any single covariance constraint may leave substantial ambiguity,
their combination can significantly reduce the set of transformations that
preserve all imposed correlations.
Such methods are often viewed as statistically robust and computationally
efficient, especially in scenarios where non-Gaussianity is weak or unreliable
but sufficient observation diversity is available.

%--------------------------------
\subsection{Structural Comparison}
Although HOS-based and SOS-based methods are often presented as alternative
strategies, both act by introducing additional invariance constraints that
restrict admissible transformations of the mixing model.
The distinction lies in the \emph{source} of constraint information:
distributional properties in the HOS case and observation diversity in the
SOS case.

From the unified viewpoint developed in Section~\ref{subsec:obs-constraints},
both families can be interpreted as imposing collections of invariance
constraints of the form~\eqref{eq:hos-invariance} or
\eqref{eq:sos-covariances}, which progressively restrict the set of admissible
transformations compatible with the observed data.
Identifiability is therefore governed by how strongly these constraints
reduce the stabilizer, rather than by the statistical order or algorithmic
realization.

Counterexamples in which individual constraint families fail to reduce the
stabilizer are summarized in Appendix~\ref{app:asymmetry}.

\section{Structural Viewpoint: Observation Constraints and Symmetry}
\label{sec:structural-viewpoint}
%--------------------------------------------
\subsection{Observation-Induced Constraints on the Mixing Model}
\label{subsec:obs_constraints_structural}

The statistical information extracted from observed mixtures imposes structural
constraints on the admissible transformations of the mixing model: only those
transformations that preserve the observed statistics remain compatible with the
data, while others lead to inconsistencies.

Here, $\Phi(H)=\{A_i(H)\}_{i\in\mathcal I}$ refers to the observation map
introduced in Section~\ref{subsec:obs-constraints}.
As defined in~\eqref{eq:global-ambiguity}, the set of transformations that leave
all constraint objects invariant forms the stabilizer
$\mathrm{Stab}(\mathcal I)$~\cite{Herstein1996AbstractAlgebra}.
This stabilizer represents the residual ambiguity of the mixing model under the
imposed observation constraints.

When only a single class of observation statistics is enforced, the corresponding
index set $\mathcal I$ is typically limited, and the stabilizer
$\mathrm{Stab}(\mathcal I)$ remains nontrivial.
Such constraints eliminate only those transformations that violate the specific
statistical properties being enforced, leaving a family of admissible
transformations that reflect intrinsic ambiguities of the separation model.

Crucially, the stabilizer $\mathrm{Stab}(\mathcal I)$ is a \emph{global} object:
it characterizes residual ambiguities at the level of exact invariance under the
observation map.
While this global characterization is conceptually clear, it does not directly
indicate how ambiguity reduction manifests under finite samples or small
perturbations.
This observation motivates the need for a complementary \emph{local}
characterization of stabilizer reduction, which is developed in the subsequent
sections via Jacobian-based sensitivity analysis.

%--------------------------------------------
\subsection{Residual Ambiguities under Multiple Constraints}
\label{subsec:residual_ambiguities}

Even when observation statistics impose nontrivial constraints on the mixing
model, residual ambiguities generally remain if only a single class of
constraints is enforced.
Such ambiguities correspond to nontrivial elements of the stabilizer
$\mathrm{Stab}(\mathcal I)$ associated with the chosen observation operators,
and cannot be ruled out based on that information alone.
These residual symmetries are not accidental; rather, they reflect intrinsic
structural properties of the separation model under limited observational
constraints.

When multiple classes of observation statistics are considered simultaneously,
each class imposes a distinct invariance on the mixing model and therefore
induces its own stabilizer.
Let $\mathcal I_1$ and $\mathcal I_2$ denote index sets associated with two
different families of observation operators.
A transformation is compatible with both constraints if and only if it preserves
both invariances.
Accordingly, the admissible transformations under multiple constraints
are characterized by the following elementary but structurally important relation.

\begin{lemma}[Intersection of stabilizers under multiple constraints]
\label{lem:stabilizer_intersection}
Let $\mathcal I_1$ and $\mathcal I_2$ be index sets corresponding to two different families
of observation operators.
Then the stabilizer associated with their joint constraint satisfies
\begin{equation}
\mathrm{Stab}(\mathcal I_1 \cup \mathcal I_2)
=
\mathrm{Stab}(\mathcal I_1) \cap \mathrm{Stab}(\mathcal I_2).
\label{eq:stabilizer-intersection}
\end{equation}
\end{lemma}

\begin{IEEEproof}
By definition of the stabilizer given in~\eqref{eq:global-ambiguity},
a transformation belongs to $\mathrm{Stab}(\mathcal I)$ if and only if it
preserves every observation operator indexed by $\mathcal I$.
Applying this definition to the union $\mathcal I_1 \cup \mathcal I_2$,
a transformation belongs to $\mathrm{Stab}(\mathcal I_1 \cup \mathcal I_2)$
precisely when it preserves all operators in both families.
This condition is equivalent to requiring membership in both
$\mathrm{Stab}(\mathcal I_1)$ and $\mathrm{Stab}(\mathcal I_2)$,
which establishes~\eqref{eq:stabilizer-intersection}.
\end{IEEEproof}

The key point is that the benefit of combining multiple constraints does not
depend on their specific algorithmic realization, but on how their induced
stabilizers intersect at a structural level.
Constraints derived from different sources of statistical information tend to
eliminate complementary classes of symmetries, even when each constraint alone
is insufficient for identifiability.

At the same time, the intersection relation in
Lemma~\ref{lem:stabilizer_intersection} remains a global statement.
It specifies which transformations are ultimately excluded, but not how such
exclusion manifests under finite perturbations or finite data.
This observation motivates the introduction of a local perspective that can
quantify stabilizer reduction in practice.

%------------------------
\subsection{Accumulation of Constraints and Structural Ambiguity Reduction}
\label{subsec:constraint-accumulation}

When multiple observation-induced constraints are enforced simultaneously,
their effects accumulate through successive restrictions of the stabilizer.
As additional constraint families are incorporated, the index set
$\mathcal I$ grows, and the stabilizer $\mathrm{Stab}(\mathcal{I})$ defined
in~\eqref{eq:global-ambiguity} is progressively reduced through repeated
intersections of the form~\eqref{eq:stabilizer-intersection}.
We refer to this progressive reduction of admissible transformations as
\emph{stabilizer shrinkage}.

From a global standpoint, stabilizer shrinkage reflects a systematic
simplification of the symmetry inherent in the mixing model under limited
observation.
Each newly imposed class of observation statistics excludes transformations
that were previously admissible, thereby reducing the degrees of freedom
associated with residual ambiguity.
This reduction is induced by the observation mechanism itself and should not be
confused with numerical regularization or estimator-dependent effects.

In practical settings, different observation statistics restrict different
subsets of the original continuous right-unitary freedom.
As a consequence, distinct constraint families may lead to comparable
ambiguity reductions even when their statistical origins differ.
This explains why HOS-based and multi-SOS-based approaches, despite relying on
different forms of statistical information, can exhibit structurally similar
identifiability behavior.

To connect this global notion of ambiguity reduction to a local and tractable
description, we explicitly exploit the fact that $\mathrm{GL}(n)$ is a Lie group
and that $\mathrm{Stab}(\mathcal I)$ forms a closed Lie subgroup, whose local
structure is characterized by its associated Lie algebra
(cf., e.g.,~\cite{Hall2015LieGroups}).
As observation constraints accumulate, stabilizer shrinkage manifests not only
as a set-theoretic reduction of admissible transformations, but also as a
reduction in the dimension of the tangent space at the identity.
From a dimensional perspective, this progressive reduction can be schematically
expressed as
\begin{equation}
\dim T_I \mathrm{Stab}(\mathcal I)
\;\downarrow\;
\text{as constraints accumulate}.
\label{eq:stab-dim-shrink}
\end{equation}
This local reduction provides a direct measure of how strongly the observation
statistics restrict infinitesimal right-unitary perturbations of the mixing
matrix.

The tangent-space viewpoint naturally bridges global symmetry reduction and
local sensitivity analysis.
In particular, it motivates a Jacobian-based characterization of local
identifiability, in which the sensitivity of the observation map to
infinitesimal reparameterizations determines whether residual ambiguities
persist locally.
This perspective forms the basis of the Jacobian analysis developed in
Section~\ref{sec:local-identifiability}.

Overall, this structural viewpoint clarifies when and how identifiability
improves as additional observation statistics are incorporated.
By framing ambiguity reduction in terms of stabilizer shrinkage and its
manifestation in the tangent space, this section provides a coherent transition
from global ambiguity analysis to the local identifiability results presented
next.

\section{Local Identifiability via Jacobian Sensitivity}
\label{sec:local-identifiability}

\subsection{Local Perturbations and Constraint Sensitivity}
\label{subsec:local-perturbations}

Identifiability in blind source separation is often discussed as a global or
asymptotic property.
In practical settings with finite samples and imperfect statistical estimates,
however, identifiability manifests through the local behavior of the model.
It is therefore informative to examine how small perturbations of the mixing
model affect the observation statistics enforced by a given set of constraints.

Let $\mathrm{Stab}(\mathcal I)$ denote the stabilizer associated with a set of
observation constraints, as defined in
Section~\ref{sec:structural-viewpoint}.
Since $\mathrm{GL}(n)$ is a Lie group, admissible transformations in a
neighborhood of the identity can be expressed through infinitesimal
perturbations.
After whitening, these reduce locally to right-orthogonal directions, motivating
the use of skew-symmetric generators.

Specifically, we consider local perturbations of the form
\begin{equation}
H(\varepsilon)
=
H \exp(\varepsilon \Omega),
\qquad
\Omega^\top = -\Omega,
\label{eq:local-perturbation}
\end{equation}
where $\varepsilon \in \mathbb{R}$ is a small scalar parameter and
$\Omega$ is a skew-symmetric matrix generating a local right action.
This parameterization preserves orthogonality after whitening and provides a
natural coordinate system for local analysis.

Under such perturbations, some directions $\Omega$ induce noticeable changes in
the observation statistics, while others leave them nearly unchanged.
The former correspond to directions strongly constrained by the imposed
statistics, whereas the latter indicate residual ambiguities that remain weakly
constrained.
Sensitivity to local perturbations therefore provides a natural lens through
which the practical strength of observation-induced constraints can be assessed.

%------------------------------------
\subsection{Jacobian as a Numerical Probe of Residual Ambiguity}
\label{subsec:jacobian-probe}

The local sensitivity of observation-induced constraints to infinitesimal
perturbations can be examined through the Jacobian of the observation map
$\Phi$ defined in \eqref{eq:obs-map}.
The Jacobian with respect to the local perturbation
\eqref{eq:local-perturbation} is defined as
\begin{equation}
J(H)
=
\left.
\frac{\partial}{\partial \varepsilon}
\Phi\!\left(H \exp(\varepsilon \Omega)\right)
\right|_{\varepsilon=0}.
\label{eq:jacobian-def}
\end{equation}

This Jacobian captures the first-order response of the observation constraints
to infinitesimal transformations generated by $\Omega$.
Locally, stabilizer shrinkage corresponds to a reduction of admissible tangent
directions, and the Jacobian therefore indicates which perturbation directions
remain compatible with the imposed constraints.

In practice, the Jacobian is used as a numerical probe of local constraint
strength.
Small Jacobian responses correspond to transformations that leave the statistics
nearly invariant, indicating weakly constrained directions associated with
residual ambiguities.
Conversely, large responses indicate directions strongly restricted by the
observation constraints.

For completeness, the formal definition of the Jacobian-based probe and its
local parameterization are summarized in Appendix~\ref{app:jacobian-probe}.

%------------------------------------
\subsection{Relation to Structural Ambiguity Reduction}
\label{subsec:structural-relation}

The numerical behavior revealed by Jacobian sensitivity corresponds directly to
the structural ambiguity reduction discussed in
Section~\ref{sec:structural-viewpoint}.
Observation-induced constraints restrict admissible transformations at the group
level through stabilizer shrinkage, and this restriction appears locally through
the behavior of the Jacobian.

Perturbation directions that remain weakly constrained by the imposed statistics
correspond to tangent directions of the stabilizer
$\mathrm{Stab}(\mathcal I)$ at the identity.
Here, $T_I\mathrm{Stab}(\mathcal I)$ denotes the tangent space of the stabilizer
$\mathrm{Stab}(\mathcal I)$ at the identity element, representing the space of
infinitesimal admissible transformations.
Such relationships between infinitesimal invariances and the tangent space of a
symmetry group are standard in Lie theory (cf., e.g.,~\cite{Hall2015LieGroups}).
When $\mathrm{Stab}(\mathcal I)$ forms a regular Lie subgroup of
$\mathrm{GL}(n)$, its tangent space at the identity is naturally identified with
the Lie algebra associated with the stabilizer, as expressed in
\eqref{eq:local_global_correspondence}.

At the level of first-order sensitivity, this relation can be summarized by
\begin{equation}
\ker J(H)
\;\subseteq\;
T_I \mathrm{Stab}(\mathcal I),
\label{eq:kernel_tangent}
\end{equation}
indicating that locally unobservable perturbations align with infinitesimal
stabilizer directions.
Under the same regularity conditions, this inclusion tightens to
$\ker J(H)=T_I\mathrm{Stab}(\mathcal{I})$.

As additional constraints are accumulated, the stabilizer shrinks and its local
tangent space decreases accordingly.
Numerically, this reduction manifests as suppression of weakly sensitive
directions in the Jacobian, reflected by the collapse of $\ker J(H)$ and the
growth of its smallest singular value.
\section{Unifying Non-Gaussianity and Observation Diversity}
\label{sec:tradeoff}

\subsection{Motivation}
\label{subsec:motivation}

Classical blind source separation methods based on higher-order statistics (HOS)
and those based on second-order statistics (SOS) are often treated as
fundamentally different.
However, from the structural viewpoint developed in
Sections~\ref{sec:structural-viewpoint}--\ref{sec:local-identifiability},
both classes of methods reduce the same object: the residual ambiguity
captured by the stabilizer of the observation constraints.

Let $\mathcal I$ denote the index set of constraint operators associated with a
given observation mechanism, and recall the definition of the stabilizer
$\mathrm{Stab}(\mathcal I)$ given in~\eqref{eq:global-ambiguity}.
Identifiability improves as this stabilizer shrinks under increasingly rich
constraint families, regardless of how the constraints are instantiated.

From this perspective, the apparent dichotomy between HOS-based and SOS-based
methods reflects different choices of the constraint index set $\mathcal I$
rather than fundamentally distinct identifiability mechanisms.
HOS-based methods correspond to constraint families
$\mathcal I_{\mathrm{HOS}}$ induced by higher-order statistics,
whereas SOS-based methods correspond to constraint families
$\mathcal I_{\mathrm{SOS}}$ induced by observation diversity.
A unified treatment therefore requires comparing how these different index sets
reduce the same stabilizer, rather than comparing algorithms in isolation.

%----------------------------------
\subsection{Numerical Identifiability via Local Sensitivity}
\label{subsec:numerical-identifiability}

While the stabilizer $\mathrm{Stab}(\mathcal I)$ provides a global
characterization of residual ambiguity, its direct computation or comparison is
generally intractable.
As discussed in Section~\ref{sec:local-identifiability}, a local characterization
is obtained by examining the Jacobian of the observation map.

Let $J(H)$ denote the Jacobian defined in~\eqref{eq:jacobian-def}.
The kernel $\ker J(H)$ characterizes infinitesimal perturbations that remain
compatible with the imposed constraints and therefore provides a local
representation of stabilizer shrinkage.
In particular, under suitable regularity conditions,
\begin{equation}
\ker J(H)
\;\cong\;
\mathfrak{lie}\!\left(\mathrm{Stab}(\mathcal I)\right),
\label{eq:local_global_correspondence}
\end{equation}
which links global symmetry reduction to local sensitivity.
Here, $\mathfrak{lie}\!\left(\mathrm{Stab}(\mathcal I)\right)$ denotes the Lie
algebra associated with the stabilizer, that is, the space of infinitesimal
generators of transformations that preserve the imposed observation constraints.

In practice, rather than using $\ker J(H)$ as a binary identifiability
criterion, we summarize local constraint strength through the smallest singular
value $\sigma_{\min}(J(H))$, which measures how strongly stabilizer directions
are suppressed.
This quantity serves as a numerical probe of stabilizer shrinkage rather than
as an estimator or optimality metric.

%----------------------------------
\subsection{Instantiating Stabilizer Shrinkage with HOS and SOS}

The unifying interpretation becomes explicit when HOS-based and SOS-based
constraints are examined within the same stabilizer framework.
For HOS-based methods, the constraint family $\mathcal I_{\mathrm{HOS}}$ is
determined by higher-order cumulant operators, whereas for SOS-based methods,
the constraint family $\mathcal I_{\mathrm{SOS}}$ is determined by collections
of second-order covariance operators.

Although $\mathcal I_{\mathrm{HOS}}$ and $\mathcal I_{\mathrm{SOS}}$ differ in
their statistical origins, both act on the same object through stabilizer
intersection:
\begin{equation}
\mathrm{Stab}(\mathcal I_{\mathrm{HOS}} \cup \mathcal I_{\mathrm{SOS}})
=
\mathrm{Stab}(\mathcal I_{\mathrm{HOS}})
\cap
\mathrm{Stab}(\mathcal I_{\mathrm{SOS}}),
\label{eq:hos_sos_intersection}
\end{equation}
as established in~\eqref{eq:stabilizer-intersection}.
Consequently, identifiability improvements obtained by increasing
non-Gaussianity or observation diversity correspond to different ways of
shrinking the same stabilizer.

This formulation clarifies why non-Gaussianity and observation diversity can be
viewed as alternative sources of constraint richness.
Different constraint families reduce the same stabilizer through different
mechanisms, while their effect on local identifiability is captured by the same
Jacobian-based sensitivity.
From a systems perspective, architectural choices (e.g., the number of lags or
sensing diversity) and statistical assumptions influence identifiability only
through the extent to which they reduce the dimension of the stabilizer,
providing a structural guideline for the design of observation mechanisms in
blind source separation.
\section{Numerical Experiments}
\label{sec:experiments}

Throughout this section, we use the Jacobian-based probe
$\sigma_{\min}(J)$ as a diagnostic tool to visualize
\emph{structural ambiguity reduction} under different observation constraints.
The absolute scale of the probe depends on the specific constraint family and is
not intended for cross-experiment comparison or for ranking HOS- and SOS-based
approaches.
All experiments are designed to validate structural predictions derived from
stabilizer shrinkage, rather than to benchmark separation algorithms.

%-----------------
\subsection{Experimental Setup and Evaluation Measures}
\label{subsec:exp-setup}

We evaluate the structural effects discussed in
Sections~\ref{sec:structural-viewpoint}--\ref{sec:local-identifiability}
using controlled numerical experiments.
The purpose of this section is \emph{not} algorithm benchmarking or performance
comparison, but to visualize how ambiguity reduction emerges through the
accumulation of observation-induced constraints.

\paragraph{Model and preprocessing.}
We consider a real-valued, linear instantaneous BSS model $y = Hs$ with an equal
number of sources and observations ($n=m$).
Unless otherwise stated, we set $n=m=3$ and whiten all mixtures using the sample
covariance.
After whitening, we select a reference mixing matrix $H_0=I$ as a canonical
representative of the equivalence class discussed in
Section~\ref{subsec:bss-model}, and evaluate identifiability locally around this
reference point.
Accordingly, the remaining ambiguity corresponds to right-orthogonal
transformations.

\paragraph{Primary probe.}
Let $\Phi(H)$ denote the collection of observation statistics introduced in
Section~\ref{subsec:obs-constraints}, evaluated under a given experimental
configuration.
Local sensitivity is probed using the right-orthogonal perturbation
parameterization introduced in Section~\ref{subsec:local-perturbations}
(cf.~\eqref{eq:local-perturbation}), together with the Jacobian defined in
Section~\ref{subsec:jacobian-probe} (cf.~\eqref{eq:jacobian-def}).
In the present experiments, the Jacobian is evaluated at the reference point
$H_0=I$, and local sensitivity is assessed via infinitesimal perturbations around
this point.

Our scalar diagnostic is defined as
\begin{equation}
\mathrm{probe}(H_0)
:= \sigma_{\min}\!\bigl(J(H_0)\bigr),
\label{eq:exp_probe}
\end{equation}
which summarizes how strongly residual stabilizer directions are suppressed
locally.
All quantities are estimated from finite samples; we report Monte Carlo means
and dispersion of $\mathrm{probe}(H_0)$.

\paragraph{Secondary index (sanity check only).}
As a secondary reference, we optionally report the Amari performance index (API)
as a sanity check~\cite{amari_new_1995}.
The API is not used to select algorithms or tune parameters; structural effects
are assessed solely through the probe.

\paragraph{Implementation.}
All experiments were implemented in Python using standard numerical libraries
(NumPy/SciPy) and executed in a CPU-based environment.
No GPU-specific operations or platform-dependent optimizations were used.

%-----------------
\subsection{Effect of Higher-Order Statistics (HOS)}
\label{subsec:exp-hos}

Table~\ref{tab:hos} summarizes the simulation settings for the HOS-based
experiments, including the choice of the reference mixing matrix $H_0$ used for
local sensitivity evaluation.
These settings are chosen to isolate the effect of source non-Gaussianity on
structural identifiability under higher-order constraints.

\begin{table}[t]
\centering
\caption{Simulation settings for the HOS experiment.}
\begin{tabular}{l c}
\hline
Number of sources/observations $n=m$ & $3$ \\
Sample length $T$ & $1\times 10^{5}$ \\
Monte Carlo trials & $50$ \\
Whitening & Yes \\
Reference mixing matrix & $H_0 = I$ \\
Source distribution & Generalized Gaussian \\
Shape parameters $p$ & $\{0.8,1.0,1.5,2.0,3.0\}$ \\
HOS constraints & Fourth-order cumulant matrices \\
\hline
\end{tabular}
\label{tab:hos}
\end{table}

\begin{figure}[t]
\centering
\includegraphics[width=\linewidth]{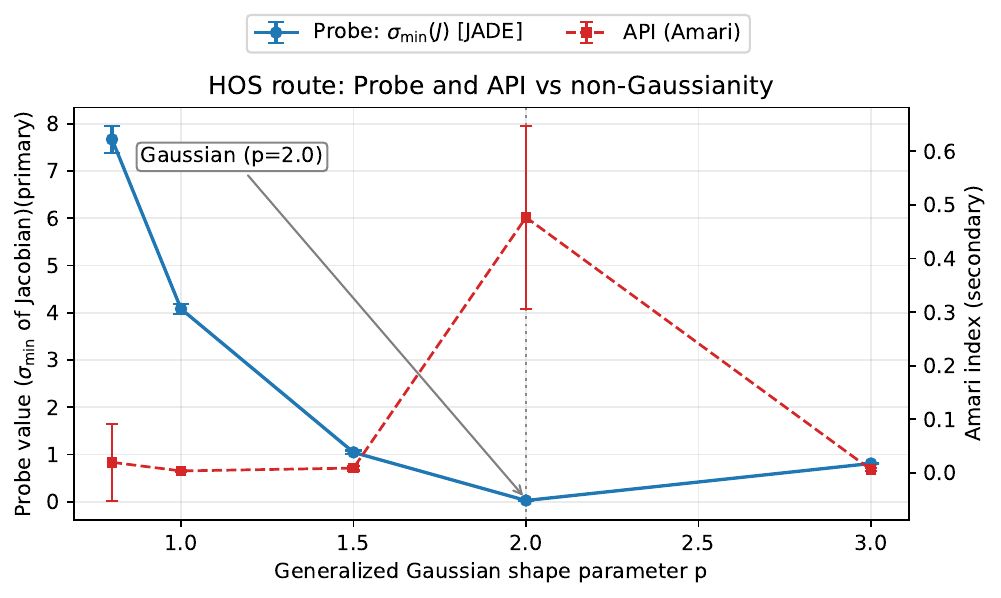}
\caption{%
Effect of source non-Gaussianity on structural identifiability in the HOS route.
Primary axis: Jacobian-based probe $\sigma_{\min}(J)$ computed from fourth-order
cumulant constraints.
The probe collapses near the Gaussian boundary ($p=2.0$) and increases for
non-Gaussian sources.
Secondary axis: Amari performance index (API).
Error bars indicate Monte Carlo dispersion.
\emph{Note that the absolute scales of both the primary probe and the secondary
separation index are experiment-dependent and should not be compared across
different figures or constraint families.}}
\label{fig:hos}
\end{figure}

Sources are i.i.d.\ generalized Gaussian, where the shape parameter $p$ controls
deviation from Gaussianity, with the Gaussian boundary at $p=2.0$.
Higher-order information is incorporated through symmetric matrices derived from
fourth-order cumulants, constructed using the standard JADE procedure.

The resulting behavior of the Jacobian-based probe evaluated at $H_0$ is shown
in Fig.~\ref{fig:hos}.
As predicted by the structural analysis, the probe collapses near the Gaussian
boundary ($p=2.0$), consistent with the vanishing of fourth-order cumulants and
the persistence of stabilizer directions under Gaussian sources.
As $p$ departs from $2.0$, the probe increases, reflecting progressive activation
of higher-order constraints and corresponding stabilizer shrinkage.
The API is shown only as a sanity check and confirms degradation of separation
quality in the Gaussian case.

%------------------
\subsection{Effect of Observation Diversity via Multi-Lag SOS}
\label{subsec:exp-sos}

We next examine ambiguity reduction driven by observation diversity using
multi-lag second-order statistics, without relying on higher-order information.
The corresponding simulation settings are summarized in
Table~\ref{tab:sos}.

\begin{table}[t]
\centering
\caption{Simulation settings for the SOS experiment.}
\begin{tabular}{l c}
\hline
Number of sources/observations $n=m$ & $3$ \\
Sample length $T$ & $1\times 10^{5}$ \\
Monte Carlo trials & $50$ \\
Whitening & Yes \\
Reference mixing matrix & $H_0 = I$ \\
Source model & Gaussian AR(1) \\
AR coefficients $(a_1,a_2,a_3)$ & $(0.2,\,0.6,\,0.9)$ \\
Lag set $\mathcal{T}_L$ & $\{1,2,\dots,L\}$ \\
Lag count $L$ & $\{1,2,3,4,5,6,7\}$ \\
SOS constraints & Lagged covariance matrices \\
\hline
\end{tabular}
\label{tab:sos}
\end{table}

\begin{figure}[t]
\centering
\includegraphics[width=\linewidth]{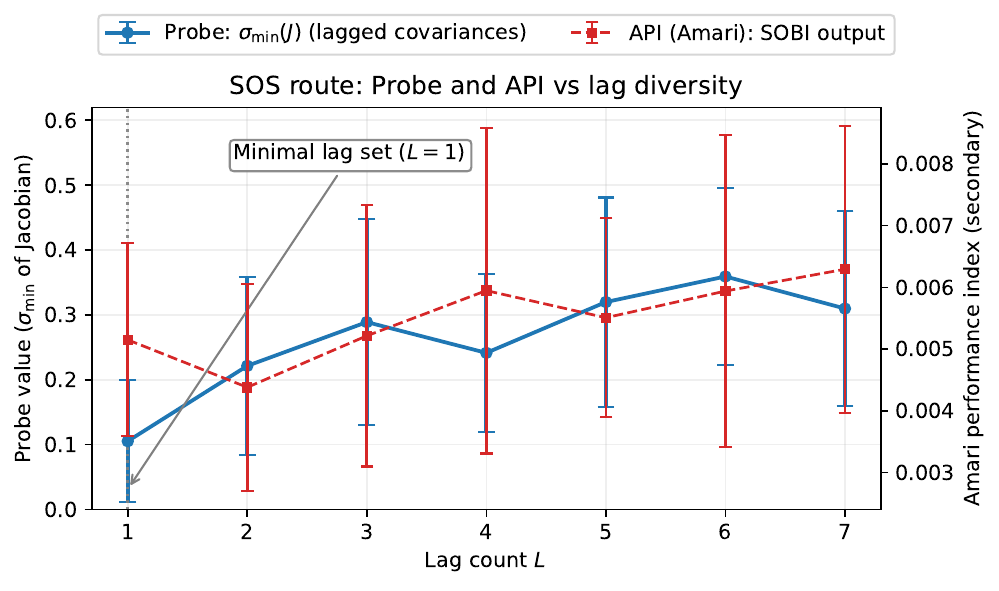}
\caption{Effect of observation diversity on structural identifiability in the
SOS route.
Primary axis: Jacobian-based probe $\sigma_{\min}(J)$ for lag counts
$L=1,2,\dots,7$.
Secondary axis: Amari performance index (API).
Increasing lag diversity strengthens the probe, while gains saturate for larger
$L$.
\emph{Note that the absolute scales of both the primary probe and the secondary
separation index are experiment-dependent and should not be compared across
different figures or constraint families.}
}
\label{fig:sos}
\end{figure}

Sources are Gaussian AR(1) processes
\[
s_i(t)=a_i s_i(t-1)+u_i(t), \qquad u_i(t)\sim\mathcal{N}(0,1),
\]
with distinct coefficients $(a_1,a_2,a_3)=(0.2,0.6,0.9)$.
Lagged covariance matrices of the whitened observations,
\[
R_z(\tau)=\mathbb{E}[z(t)z(t-\tau)^\top], \qquad \tau\in\mathcal{T}_L,
\]
constitute the SOS constraints.

Figure~\ref{fig:sos} reports the resulting probe values evaluated at the
reference mixing matrix $H_0$ as a function of the number of lags $L$.
For the minimal lag set ($L=1$), the probe remains small, indicating weak
constraint strength.
As additional lags are included ($L>1$), the probe increases overall, reflecting
progressive suppression of residual stabilizer directions through accumulated
temporal constraints.
For larger $L$, gains diminish due to finite-sample effects.
The API remains nearly constant, indicating that separation quality is already
saturated, while structural ambiguity continues to shrink.

%------------------
\subsection{Trade-off Between Source Statistics and Observation Diversity}
\label{subsec:exp-tradeoff}

Finally, we visualize the trade-off between source statistics and observation
diversity using second-order observation constraints, without explicitly
combining higher-order statistics.
Here, the shape parameter $p$ characterizes the source process only and does not
enter the SOS constraint operators themselves; the probe
$\mathrm{probe}_{\mathrm{SOS}}$ is computed exclusively from second-order
(lagged covariance) operators, while varying $p$ modifies the effective
second-order structure available after whitening.
For each $(p,L)$, we evaluate the Jacobian-based probe defined
in~\eqref{eq:exp_probe} for the SOS constraint family:
\begin{equation}
\mathrm{probe}_{\mathrm{SOS}}(p,L)
\;:=\;
\sigma_{\min}\!\bigl(J_{\mathrm{SOS}}(H_0)\bigr).
\label{eq:probe_sos}
\end{equation}
Based on this quantity, we define the minimal lag count required to reach a
target sensitivity level $\varepsilon$ as
\begin{equation}
L_{\min}(p;\varepsilon)
\;:=\;
\min\{\,L:\ \mathrm{probe}_{\mathrm{SOS}}(p,L)\ge \varepsilon\,\}.
\label{eq:Lmin}
\end{equation}

The resulting trade-off is summarized in Fig.~\ref{fig:tradeoff_sos}.
Near the Gaussian boundary, where statistical structure is weak, larger lag
diversity is required to reach the same probe level.
Away from this boundary, fewer lags suffice.
Different pairs $(p,L)$ that yield the same probe value therefore correspond to
structurally equivalent configurations in terms of local identifiability,
illustrating how non-Gaussianity and observation diversity contribute
comparably to stabilizer shrinkage.

\begin{figure}[t]
\centering
\includegraphics[width=\linewidth]{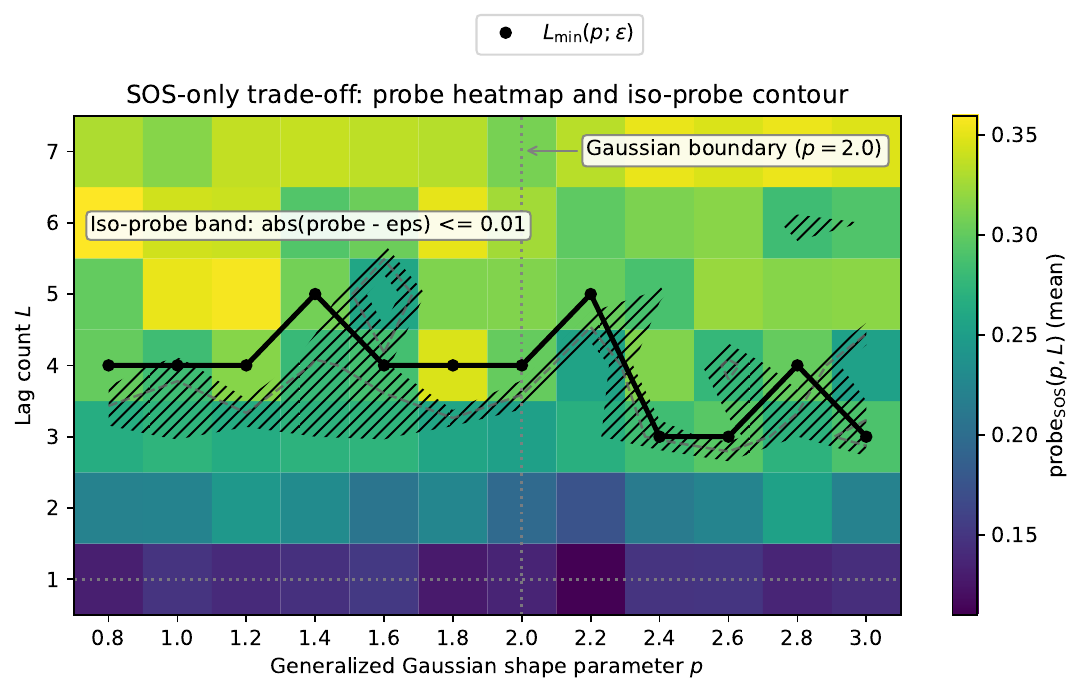}
\caption{Trade-off between source statistics and observation diversity under
second-order observation constraints.
Heatmap: $\mathrm{probe}_{\mathrm{SOS}}(p,L)$.
Markers: $L_{\min}(p;\varepsilon)$ defined in~\eqref{eq:Lmin}.
This visualization highlights structurally equivalent configurations rather
than performance optima.}
\label{fig:tradeoff_sos}
\end{figure}
\section{Conclusion}
\label{sec:conclusion}

This paper examined identifiability in blind source separation (BSS) from a
structural standpoint, focusing on how observation-induced constraints eliminate
symmetry-induced ambiguities.
HOS-based and SOS-based approaches were characterized as mechanisms that restrict
the class of transformations consistent with the observed data.
Within this formulation, identifiability is governed by the extent to which
statistical constraints shrink residual ambiguities, and the classical
distinction between HOS and SOS is attributed to the origin of the imposed
constraints.

To connect this structural characterization with finite-sample behavior, we
introduced a Jacobian-based sensitivity probe as a numerical diagnostic of local
identifiability.
Numerical experiments confirmed the structural interpretation: in the HOS route,
the probe collapses near the Gaussian limit and increases with stronger
non-Gaussian constraints; in the SOS route, increasing lag diversity strengthens
the probe even under Gaussian sources; and their trade-off can be visualized
within a common structural scale.
These results provide a local characterization of residual ambiguities through
the kernel and sensitivity of the Jacobian.

Beyond individual mechanisms, the main contribution of this work is a
constraint-centered reorganization of identifiability.
By separating structural effects from algorithmic implementations, the proposed
view clarifies why methods with different formulations often exhibit similar
empirical behavior and why additional constraints can improve robustness even
when conventional performance indices saturate.
From a circuits-and-systems perspective, observation architecture choices
(e.g., lag diversity) and statistical assumptions influence identifiability
through their effect on stabilizer reduction under finite-sample and
computational constraints.
In this sense, identifiability can be viewed as the collapse of an infinite
equivalence class to a finite residual ambiguity under observation-induced
constraints (cf.~Sections~\ref{subsec:obs-constraints} and
\ref{subsec:constraint-accumulation}).

Several extensions naturally follow.
The same structural analysis applies to complex-valued formulations and to
broader observation models, including convolutive and block or oversampled
settings, as well as to multi-way representations arising in tensor-based
formulations~\cite{le_tensor_2024}.
The Jacobian-based probe also supports a design-oriented evaluation of constraint
configurations based on their ambiguity-suppressing effect.
More generally, the proposed interpretation provides a basis for extending
constraint-induced ambiguity reduction to other blind inference problems beyond
classical BSS.

Overall, by organizing identifiability around constraint-induced ambiguity
reduction, this work clarifies the operating principles of classical BSS methods
and provides a concise structural basis for the design and evaluation of future
blind separation approaches.

% ===== Appendix =====
\appendices
\appendices
\section{Formal Definition of the Jacobian-Based Probe}
\label{app:jacobian-probe}

\subsection{Observation Relations and Parameterization}
\label{app:observation-relations}

This appendix formalizes the Jacobian-based probe used in
Sections~\ref{sec:local-identifiability} and~\ref{sec:experiments}.
All models, assumptions, and notation are identical to those introduced in
Section~\ref{sec:problem-formulation}, and the local perturbation model follows
Section~\ref{sec:local-identifiability}.
The purpose here is purely notational: to state the probe definition used in the
numerical evaluations, without introducing additional theoretical claims.

Let $\Phi(H)$ denote the observation map defined in~\eqref{eq:obs-map}, whose
components depend on the chosen constraint family (e.g., HOS or multi--SOS).
Throughout, local sensitivity is evaluated under right-orthogonal
reparameterizations as in~\eqref{eq:local-perturbation}.

\subsection{Jacobian of Observation Relations}
\label{app:jacobian-definition}

Consider the local perturbation $H(\theta)=H\exp(\theta\Omega)$ with
$\Omega^\top=-\Omega$ as in~\eqref{eq:local-perturbation}.
The Jacobian of the observation map along this perturbation is defined by
\begin{equation}
J(H)
=
\left.
\frac{\partial}{\partial \theta}\,
\Phi\!\left(H \exp(\theta \Omega)\right)
\right|_{\theta = 0},
\end{equation}
which coincides with the definition used in~\eqref{eq:jacobian-def}.
The Jacobian provides a first-order linearization of how the enforced observation
statistics respond to infinitesimal right-orthogonal perturbations.

It is emphasized that $J(H)$ is used only as a local sensitivity probe.
It is not associated with an estimator, likelihood function, or statistical
optimality criterion, and is introduced solely to quantify how strongly the
imposed observation constraints respond to parameter perturbations.

\paragraph*{Concrete instantiations of $\Phi(H)$.}
The abstract observation map $\Phi(H)$ specializes to concrete forms depending
on the constraint family.
For HOS-based constraints,
\begin{equation}
\Phi_{\mathrm{HOS}}(H)
=
\bigl\{ C^{(4)}_k(H) \bigr\}_{k=1}^{K},
\end{equation}
where $C^{(4)}_k(H)$ denote symmetric matrices constructed from fourth-order
cumulants (e.g., as in the JADE construction).
For SOS-based constraints with lag set $\mathcal T_L=\{1,\dots,L\}$,
\begin{equation}
\Phi_{\mathrm{SOS}}(H)
=
\bigl\{ R_z(\tau) \mid \tau \in \mathcal T_L \bigr\},
\qquad
R_z(\tau)=\mathbb{E}[z(t)z(t-\tau)^\top].
\end{equation}
In both cases, we write
$\Phi(H)=\{A_i(H)\}_{i\in\mathcal I}$
where $A_i(H)$ denotes the $i$-th matrix-valued component of the observation
map (e.g., $A_i(H)=C^{(4)}_k(H)$ for HOS-based constraints or
$A_i(H)=R_z(\tau)$ for SOS-based constraints).
For local sensitivity analysis, we use the stacked vector representation
$\phi(H)$ obtained by vectorizing and concatenating these components:
\begin{equation}
\phi(H)
=
\begin{bmatrix}
\mathrm{vec}\,A_1(H)\\
\vdots\\
\mathrm{vec}\,A_{|\mathcal I|}(H)
\end{bmatrix}
\in\mathbb{R}^{|\mathcal I|\,n^2},
\end{equation}
where $\mathrm{vec}(\cdot)$ denotes the standard column-wise vectorization
of a matrix.
The Jacobian $J(H)\in\mathbb{R}^{(|\mathcal I|\,n^2)\times d}$ with
$d=n(n-1)/2$ is then obtained by collecting directional derivatives of
$\phi(H)$ with respect to infinitesimal right-orthogonal perturbations of $H$,
parameterized by a basis of skew-symmetric generators $\Omega$.

\subsection{Definition and Interpretation of the Probe}
\label{app:probe-definition}

The Jacobian-based probe is defined as
\begin{equation}
\mathrm{probe}(H) := \sigma_{\min}\!\left( J(H) \right),
\end{equation}
where $\sigma_{\min}(\cdot)$ denotes the smallest singular value.
Small values of $\mathrm{probe}(H)$ indicate the presence of perturbation
directions that leave the enforced statistics nearly unchanged, whereas larger
values indicate stronger suppression of such weakly observable directions.
This interpretation is consistent with the structural viewpoint in
Section~\ref{sec:structural-viewpoint} and the local analysis in
Section~\ref{sec:local-identifiability}.

Importantly, the probe does not define identifiability in a strict sense, nor
does it characterize the full symmetry structure of the model.
It is used in Section~\ref{sec:experiments} only as a diagnostic tool to
visualize and compare the effects of different constraint configurations, without
introducing new identifiability criteria.

\subsection{Practical Computation from Sample Statistics}
\label{app:practical-computation}

In experiments, the statistics entering $\Phi(H)$ are estimated from finite data
samples.
Accordingly, $J(H)$ and $\mathrm{probe}(H)$ are computed using sample-based
estimates, and the reported values reflect both the imposed structural constraints
and estimation variability.
All numerical results in Section~\ref{sec:experiments} are therefore averaged over
Monte Carlo trials and accompanied by dispersion measures.

%-----------------------
\section{Situations in Which Identifiability Constraints Do Not Shrink the Stabilizer}
\label{app:asymmetry}
%-----------------------

This appendix provides explicit examples that support the asymmetry stated in
Section~\ref{subsec:obs-constraints} (Remark on identifiability):
``richness'' alone does not shrink the stabilizer unless constraints act at the
appropriate level.
All notation and assumptions follow Section~\ref{sec:problem-formulation} and
Section~\ref{subsec:obs-constraints}, and we focus only on the transformations
that preserve the available observations.

%-----------------------
\subsection{Output Non-Gaussianity Alone Is Insufficient}
\label{app:output_nongaussian}

Assume that only the output distribution $P_y$ is observed, while no constraint
is imposed on the source distribution $P_s$.
Then the right-unitary reparameterization
\begin{equation}
H \mapsto H U, \qquad U \in \mathcal{U}(n),
\label{eq:right_unitary_action}
\end{equation}
cannot be excluded from $P_y$ alone.
Indeed, defining $s' = U^{\ast}s$ yields the equivalent factorization
\begin{equation}
y = Hs = (HU)s',
\label{eq:equivalent_factorization}
\end{equation}
so the same $P_y$ is compatible with an entire right-unitary orbit of mixing
matrices.
Hence, output non-Gaussianity by itself does not restrict the stabilizer of the
mixing model.

% NOTE: \label{eq:bss_model} and \label{eq:source_rotation} are kept as-is to
% avoid changing labels, although they are not essential in the shortened text.

%-----------------------
\subsection{Input-Side Structural Diversity Without Source Constraints Is Also Insufficient}
\label{app:input_structural}

Assume that the source process is Gaussian but exhibits rich second-order
structure (e.g., nontrivial multi-lag covariances, oversampled/block covariances,
or cyclostationary second-order statistics).
Without additional source-side constraints that fix a preferred representation,
the same right-unitary reparameterization persists:
\begin{equation}
s'(t) = U^{\ast}s(t), \qquad H' = H U, \qquad U \in \mathcal{U}(n).
\label{eq:unitary_reparam}
\end{equation}
Gaussianity is preserved under unitary transformations, and second-order
structure transforms by unitary congruence; for example,
\begin{equation}
\Sigma_{s'}(\tau)
:= \mathbb{E}\!\left[s'(t) s^{\ast}(t-\tau)\right]
= U^{\ast}\Sigma_s(\tau)U .
\label{eq:unitary_covariance_congruence}
\end{equation}
Thus, input-side second-order ``richness'' alone does not shrink the stabilizer
unless additional constraints anchor the source representation.

%-----------------------
\subsection{Failure of Single-Lag SOS}
\label{app:single_lag_counterexample}

This subsection presents a minimal counterexample that isolates the limitation
of a single \emph{output-side} second-order constraint and clarifies why
observation diversity is essential for stabilizer shrinkage.
All notation and assumptions follow
Sections~\ref{sec:problem-formulation}--\ref{subsec:obs-constraints}, and we focus
exclusively on the effect of lagged covariance operators.

Accordingly, when only the single-lag covariance operator is enforced, the
stabilizer remains the full unitary group,
\begin{equation}
    \mathrm{Stab}(H;\mathrm{Cov}(0)) = \mathcal{U}(n),
\end{equation}
and the corresponding Jacobian kernel is nontrivial,
\begin{equation}
    \ker J(H) \neq \{0\},
\end{equation}
where $J(H)$ denotes the Jacobian of the observation map defined in
\eqref{eq:jacobian-def}.
This demonstrates that a single second-order constraint does not reduce the
continuous right-unitary ambiguity: infinitesimal right-unitary perturbations
remain locally invisible to the enforced statistics.

Now consider the addition of a second lag $\tau \neq 0$, yielding the pair of
operators $\{\Sigma_y(0), \Sigma_y(\tau)\}$.
Under generic conditions---specifically, when the joint spectrum of the associated
source covariance operators is simple---the joint commutant of these operators is
finite.
In this case, the stabilizer collapses to a finite set, and the local ambiguity
vanishes:
\begin{equation}
    \ker J(H) = \{0\}.
\end{equation}
Thus, the inclusion of a second lag is sufficient to eliminate the continuous
right-unitary freedom that persists under a single-lag constraint.

This example highlights the minimal nature of stabilizer shrinkage under
\emph{output-side} second-order statistics: observation diversity must be
sufficient to break the relevant symmetry.
Single-lag SOS fails to do so, while multi-lag SOS succeeds precisely by inducing
a finite stabilizer through joint constraints.

%-----------------------
\subsection{Interpretation}
\label{app:interpretation}

The above examples support the central message of the paper:
identifiability in BSS emerges only when constraints restrict the admissible
transformations of the mixing model at the appropriate level.
In particular, stabilizer shrinkage requires either
(i) source-side constraints that fix a preferred representation (e.g.,
non-Gaussianity combined with independence), or
(ii) output-side operator diversity that effectively constrains the right-unitary
action (e.g., multi-lag/oversampled covariance operators whose joint commutant is
finite).
By contrast, statistical or structural ``richness'' by itself does not shrink the
stabilizer unless anchored by such constraints.

\subsection{HOS Trade-off Visualization}
\label{subsec:arxiv-hos-tradeoff}

For completeness, we provide an additional trade-off visualization for the HOS route.
We construct a family of fourth-order cumulant matrices $\{C_i\}$ (JADE-style) from whitened mixtures and evaluate the Jacobian-based probe at $H=I$.
Let $K$ denote the number of cumulant matrices used as constraints (sorted by Frobenius norm), and define
\[
\mathrm{probe}_{\mathrm{HOS}}(p,K)=\sigma_{\min}\!\bigl(J_{\mathrm{HOS}}\bigr).
\]
Figure~\ref{fig:hos_tradeoff_arxiv} shows a heatmap of $\log_{10}\mathrm{probe}_{\mathrm{HOS}}(p,K)$, together with the design-oriented summary
\[
K_{\min}(p;\varepsilon):=\min\{\,K:\ \mathrm{probe}_{\mathrm{HOS}}(p,K)\ge \varepsilon\,\},
\]
where $\varepsilon=0.5$.

As expected, the HOS probe collapses rapidly near the Gaussian boundary ($p=2$) across all $K$, reflecting the disappearance of fourth-order information for Gaussian sources.
Away from the boundary, increasing $K$ strengthens the probe, and only a small number of HOS constraints may suffice to exceed the target level.
The hatched region indicates an iso-probe band
\[
|\mathrm{probe}_{\mathrm{HOS}}(p,K)-\varepsilon|\le \delta
\]
(with $\delta=0.05$), serving as a reading aid rather than a sharp boundary.

\begin{figure}[t]
\centering
\includegraphics[width=\linewidth]{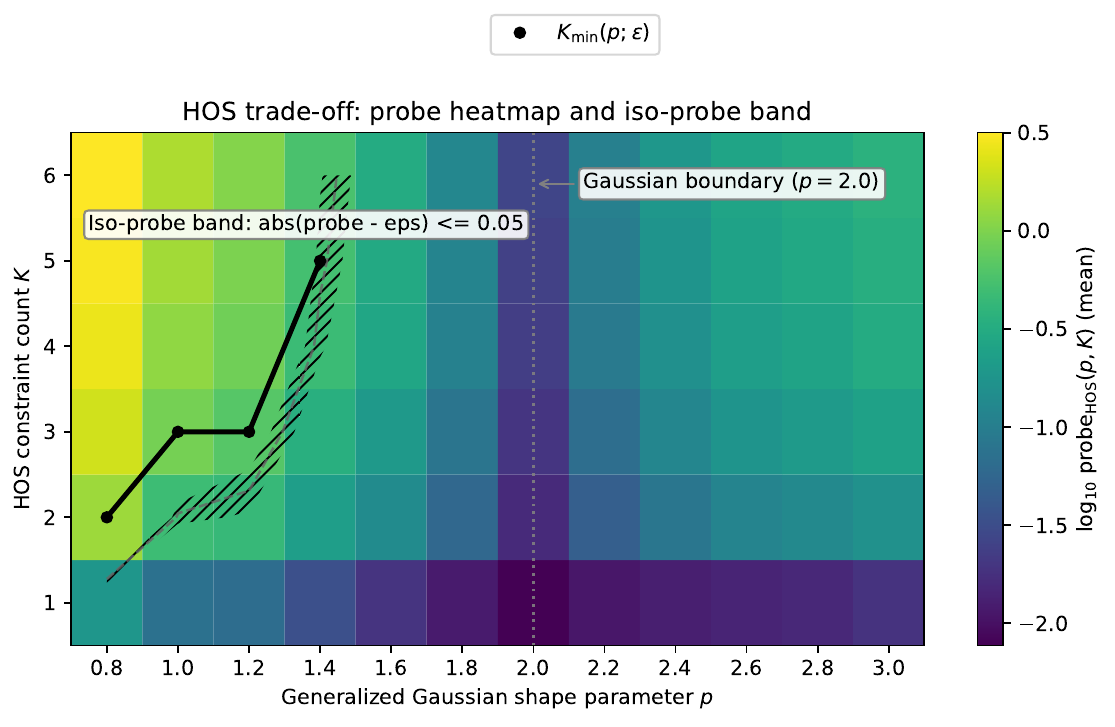}
\caption{HOS trade-off visualization.
Heatmap: $\log_{10}\mathrm{probe}_{\mathrm{HOS}}(p,K)$ (Monte Carlo mean) as a function of generalized Gaussian shape parameter $p$ and the number $K$ of cumulant-matrix constraints (JADE-style, sorted by Frobenius norm).
Black markers and solid line: $K_{\min}(p;\varepsilon)$ with $\varepsilon=0.5$.
Hatched region: iso-probe band $|\mathrm{probe}_{\mathrm{HOS}}-\varepsilon|\le\delta$ with $\delta=0.05$ (reading aid; not a sharp boundary).
The vertical dotted line marks the Gaussian boundary ($p=2.0$).}
\label{fig:hos_tradeoff_arxiv}
\end{figure}

\section{Remarks on Classical Algorithms (arXiv version)}
\label{app:remarks-classical}

\subsection{AMUSE and FastICA}

Algorithms such as AMUSE and FastICA are often described as SOS- or
HOS-based methods. From the viewpoint of identifiability, AMUSE relies on a
single nonzero-lag covariance and therefore corresponds to the minimal SOS
configuration ($L=1$), while FastICA exploits higher-order information through
an implicit optimization criterion. In both cases, the set of identifiable
transformations is determined by the statistics being enforced; the algorithms
themselves primarily affect estimation accuracy under finite samples rather
than the underlying structural identifiability.

% ===== References =====
\bibliographystyle{IEEEtran}
%\bibliography{cbsp-1st_arXiv/CBSP-1st-reference}
\bibliography{CBSP-1st-reference}

% \newpage

 

% \end{IEEEbiography}



% \vfill

\end{document}